\documentclass[12pt,leqno]{amsart}
\usepackage{pgf,graphicx}
\usepackage{graphicx}
\usepackage{amsmath,amssymb}
\usepackage{multimedia}
\usepackage{cancel}
\usepackage[english]{babel} 
\usepackage[T1]{fontenc}
\usepackage{textcomp}
\usepackage[utf8]{inputenc} 

\usepackage{cite}

\usepackage{vmargin}
\usepackage{amssymb, amsmath, amsbsy,amsthm} 
\usepackage{multicol}

\usepackage{times}

\newcommand{\authorfootnotes}{\renewcommand\thefootnote{\@fnsymbol\c@footnote}}

\numberwithin{equation}{section}
\newtheorem{theorem}{Theorem}[section]
\newtheorem{remark}[theorem]{Remark}
\newtheorem{definition}[theorem]{Definition}
\newtheorem{corollary}[theorem]{Corollary}

\begin{document}

\title[Consistency of extended Nelson-Siegel curves]{Consistency of extended Nelson-Siegel curve families with the \\ Ho-Lee and Hull and White short rate models}
\author{Patricia Kisbye}
\dedicatory{FaMAF. Universidad Nacional de Córdoba}

\thanks{Partially supported by SeCyT-UNC, code number 30720150100227CB}
\thanks{Address correspondence to Patricia Kisbye, FaMAF, Universidad Nacional de Córdoba;\\ e-mail: patricia.kisbye@unc.edu.ar}

\author{Karem Meier}
\thanks{Partially supported by SeCyT-UNC, code number 30720150100227CB; and CONICET grant.}
\thanks{Address correspondence to Karem Meier, FaMAF, Universidad Nacional de Córdoba;\\ e-mail: kam0107@famaf.unc.edu.ar}

%
\date{}
\thanks{This paper is in final form and no version of it will be submitted for
publication elsewhere.}


\begin{abstract}
    Nelson and Siegel curves are widely used to fit the observed term structure of interest rates in a particular date. By the other hand, several interest rate models have been developed such their initial forward rate curve can be adjusted to any observed data, as the Ho-Lee and the Hull and White one factor models. In this work we study the evolution of the forward curve process for each of this models assuming that the initial curve is of Nelson-Siegel type. We conclude that the forward curve process produces curves belonging to a parametric family of curves that can be seen as extended Nelson and Siegel curves.
 
\end{abstract}

\maketitle

\section{INTRODUCTION}

A standard procedure when dealing with concrete interest rate models is to calibrate the initial forward curve with the market observed data. That is the case of the Ho and Lee, and Hull and White models, where every curve can be perfectly fitted by adjusting the model parameters.
By the other hand, some parametric curves are extendedly used to fit daily data, as is the case of the Nelson-Siegel curves \cite{ns}
\begin{equation}
    \label{nelson-siegel}
f_{NS}(\tau)=z_1+z_2\,e^{-\lambda\tau}+z_3\tau\,e^{-\lambda\tau}, \qquad \tau  \ge 0 ,
\end{equation}

with $z_1$, $z_2$, $z_3$ and $\lambda$ being specified parameters. So this means that it is possible to choose the Ho-Lee and Hull-White models parameters in such a way that the initial forward rate curve fits with an specific Nelson and Siegel curve. 
In this work we show that in this particular case, the following forward rate curves moves on a manifold generated by specific parametric forward curves that can be written as a sum of a Nelson and Siegel curve and a linear or an exponential function, depending upon the short rate model.
A Nelson-Siegel curve can be decomposed in three factors: $1$, $e^{-\lambda \tau}$ and $\tau\,e^{-\lambda \tau}$. The constant factor is related with the long term interest rate level. The exponential decay is the second factor, with an upward slope if $z_2>0$ or downward if $z_2<0$. The third factor gives a hump or o a trough, depending on $z_3$. Finally, $\lambda$ is called the \textit{shape parameter}, and it determines the critical point of the third factor and the steepness of the hump/trough. (See \cite{jan}).

We prove that the forward rate curves produced by the Ho-Lee model and the Hull and White models when starting with a Nelson and Siegel curve is decomposed in four factors. Three of them are the same as in the Nelson and Siegel curve, and the fourth is a linear function $(\tau)$ in the Ho-Lee model or an exponential function ($c_1\,e^{-a\tau}+c_2\,e^{-2a\tau}$) in the Hull and White model, where $a$ is a model parameter. 
This result extends part of Bjork and Christensen (see \cite{bc}) paper results, where they proved that the above two models are  inconsistent with a strictly Nelson-Siegel manifold.

In sections \S 3 and \S 4, we present the Ho-Lee and Hull-White models, and derive the formula for the corresponding forward curve. In each case, we choose a Nelson and Siegel curve as the initial forward curve and then prove that the following ones are extended Nelson-Siegel curves in the sense that they can be written as (\ref{nelson-siegel}) plus a linear function or an exponential function.

In particular we also prove that each of these two short rate models are consistent with a forward curve manifold $\mathcal G^\lambda$, for each $\lambda>0$.

\subsection{Notation and facts}
In this section we assume a probability space $(\Omega, \mathcal F, Q)$. Let $W(t)$, $t\ge 0$ be a Wiener process, and $\{\mathcal F_t\}_{t\ge 0}$ be the filtration generated by $W(t)$. Let $(\Omega, \mathcal F, Q, \{F_t\}_{t\ge 0})$ denote the filtered probability space. A stochastic process $\alpha$ is called an \textit{adapted process} if $\alpha(t)$ is $\mathcal F_t$-measurable for every $t\ge 0$. An Ito process with \textit{drift} $\mu$ and \textit{volatility} $\sigma$ is a stochastic process $X(t)$ such that 
\begin{equation}
X(t)=X(0)+\int_0^t \mu(s,X(s))\,ds + \int_0^t \sigma(s,X(s))\,dW(s), \qquad t \ge 0, \label{ito1}
\end{equation} where $\mu$ and $\sigma$ are adapted process and the second integral at the right hand side is an Ito integral, \cite{oks}.
Equation (\ref{ito1}) is usually written in terms of a \textit{stochastic differential equation}
\begin{equation}\label{ito2}
dX(t)=\mu(t,X(t))dt + \sigma(t,X(t))dW(t).
\end{equation} 

We also introduces the relationship between the Ito and the Stratonovich integral forms. If $X$ is an Ito process as in  (\ref{ito2}), then its Stratonovich integral form is as follows,
\begin{equation}
    \label{Stratonovich}
    dX(t)=(\mu(t,X(t)) +\phi(t,X(t)))  \,dt + \sigma(t,X(t))\circ dW(t),
\end{equation}
where $\circ$ denotes the Stratonovich integration an  $\phi(t,X(t))$ is a quadratic co-variance term, 
(see \cite{oks}). If $\sigma$ in (\ref{ito2}) is deterministic, then $\phi(t,X(t))=0$.

We assume the existence of zero coupon bond market $\{P(t,T), \ 0 \le t \le T\}$, where  $P(t,T)$ denotes the price  at time $t$ of a bond with \textit{maturity} $T$. We call this bond a $T$-bond. We assume that for each $t$, the curve $T\mapsto P(t,T)$ is differentiable, with positive values and that $\{P(t,T),\ t\ge 0\}$ follows an Ito process for each $T  \ge 0 $. The \textit{forward rate curve} associated to this bonds is given by
$$f(t,T)=-\frac{\partial \ln P(t,T)}{\partial T},$$
and the short interest rate is given by $r(t)=f(t,t)$. We shall suppose an arbitrage free model and we denote with $Q$ the corresponding martingale measure. Every Ito process shall be described in terms of the $Q$-measure. As in the Heath, Jarrow y Morton (HJM) \cite{hjm} framework, we assume that the forward rate curve dynamics is given by a family of stochastic differential equations, which expression under $Q$ is given by
\begin{equation}
df(t,T)=\alpha(t,T)\,dt + \sigma_0(t,T)\,dW(t),
\label{proceso-dF}
\end{equation}
with $\alpha$ and $\sigma_0$ adapted processes. The hypothesis of an arbitrage free market implies the HJM-drift condition on $\alpha$. More precisely
\begin{equation}\alpha(t,T)=\sigma_0(t,T)\int_t^T \sigma_0'(t,s)\,ds \label{condicion-tendencia}, \end{equation}
where the superscript in $\sigma'$ denotes transpose in case of a vectorial process. 

Given a $T$-bond, we denote  $\tau=T-t$ the time \textit{up to maturity} of the bond. The Brace and Musiela parametrization \cite{bm} describes the forward curve process in terms of $t$ and $\tau$ as follows,
\begin{equation}
    f_r(t,\tau)=f(t,t+\tau),
    \label{r-forward}
\end{equation}
and so  $r(t)=f_r(t,0)$.
Under this parametrization, equation (\ref{proceso-dF}) can be written as
\begin{equation}
df_r(t,\tau)=\left(\frac{\partial}{\partial \tau}f_r(t,\tau) +\sigma(t,\tau)\int_0^\tau \sigma'(t,s)\,ds \right)\,dt+\sigma(t,\tau)\,dW(t),
\label{proceso-df}
\end{equation}
where $\sigma(t,\tau)=\sigma_0(t,t+\tau)$.

In  particular, the forward rate process (\ref{proceso-df}) can be expressed in terms of the Stratonovich integral form as
\begin{eqnarray}
df_r(t,\tau)&=&\left(\frac{\partial }{\partial \tau}f_r(t,\tau) +\sigma(t,\tau)\int_0^\tau \sigma'(t,s)\,ds +
\phi(t,\tau)
\right)\,dt  \notag \\&&\quad + \sigma(t,\tau)\circ dW(t),
\label{proceso-df-Stratonovich}
\end{eqnarray}

\subsection{Consistency}
Consistency between short rate models and forward curves manifolds were stated by Björk and Christensen in \cite{bc}. To make this work more self contained, we recall some of their definitions and the main theorem.
Let $\mathcal M$ a given one factor interest rate model specifying a forward rate process $f_r(t,\cdot)$. In terms of the Musiela parametrization, $f_r$ satisfies a stochastic differential equation:
\begin{equation}
    \label{general-forward}
df_r(t,\tau)=\left(\frac{\partial}{\partial \tau}f_r(t,\tau)+\alpha(t,\tau)\right)\,dt + \sigma(t,\tau)dW(t), 
\end{equation}
$t\ge 0$, $\tau  \ge 0 $, where $\alpha$ and $\sigma$ are adapted processes. In particular, the no arbitrage Heath, Jarrow and Morton (HJM) drift condition implies that $\alpha(t,x)=\sigma(t,x)\int_t^x \sigma(t,s)\,ds$. Also, using the Stratonovich integral form, (\ref{general-forward}) can be written as
\begin{eqnarray}
    df_r(t,\tau)&=& \left(\frac{\partial}{\partial \tau}f_r(t,\tau)+\sigma(t,x)\int_t^{\tau} \sigma(t,s)\,ds+\phi(t,\tau)\right)\,dt \notag  \\& &\quad + \sigma(t,\tau)\circ dW(t).\label{stratonovich-forward}
\end{eqnarray}
Björk and Christensen \cite{bc} stated the following definitions of consistency between a short rate model and a parametric family of curves. At first, let $\mathcal Z\subseteq \mathbb R^d$ be a set of parameters, and let $G:\mathcal Z\mapsto C[0,\infty)$ be a smooth function. The \textit{forward curve manifold} $\mathcal G$ is defined as $\mathcal G=\mbox{Im}(G)$. That is,
$$\mathcal G=\{G(\cdot;z):[0,\infty)\mapsto \mathbb R\},$$
where, with some abuse of notation $G(\cdot;z)$ denotes the function $G(z)$. 

\begin{definition}(invariance). 
Consider a given interest rate model $\mathcal M$, specifying a forward rate process $f_r(t, \cdot)$, and a forward curve manifold $\mathcal G$. We say that $\mathcal G$ is \textit{invariant under the action of $f_r$} if, for every fixed initial time $s$, the condition $f_r(s, \cdot ) \in \mathcal G$ implies that $f_r(t,\cdot) \in \mathcal G$, for all $t\ge s$, a.s.
\end{definition}
Björk and Christensen also stated a more restricted concept of invariance, the $f_r$-invariance.
\begin{definition}
($f_r$-invariance). Consider a given interest rate model $\mathcal M$, specifying a forward rate process $f_r(t, \cdot)$ as in (\ref{stratonovich-forward}), as well as a forward curve manifold $\mathcal G$. We say that $\mathcal G$ is $f_r$-invariant under the action of the forward rate process
$f_r(t, \cdot)$ if there exists a stochastic process $Z$ with state process $\mathcal Z$ and possessing
a Stratonovich differential of the form
$dZ(t) = \gamma (t, Z(t))dt + \psi (t, Z(t)) \circ dW (t),$
such that, for every fixed choice of initial time $s$, whenever $y(s, \cdot )\in \mathcal G$, the
stochastic process defined by
$$y(t, \tau ) = G (\tau ; Z(t)), \qquad \forall t \ge s,\  x \ge  0,$$
satisfies the SDE (\ref{stratonovich-forward})
 with initial condition $f_r(s, \cdot) = y(s, \cdot)$.

\end{definition}
In this case, we say that the short rate model $\mathcal M$ and the manifold $\mathcal G$ are \textit{consistent}. 
Is easy to prove that $f_r$-invariance implies invariance. Moreover, Björk and Christensen proved the following theorem.
\begin{theorem}
\label{theorem-consistency}
The forward curve manifold $\mathcal G$ is $f_r$-invariant for the forward rate process $f_r(t, \cdot)$ in $\mathcal M$ if and only if
\begin{eqnarray}
G_\tau (\cdot; z) + \sigma (t, \cdot)\int_0^{(\cdot)} \sigma'(t,s)\,ds + \phi(t,\cdot) &\in& \mbox{Im }[G_z(\cdot;z)] \label{drift-condition}\\
\sigma(t,\cdot)&\in& \mbox{Im }[G_z(\cdot;z)] \label{volatility-condition}
\end{eqnarray}
for all $(t, z) \in \mathbb [0,\infty)\times \mathcal Z$.
$G_\tau$ and $G_z$ denote the Frechet derivatives of $G$ with respect to $\tau$ and $z$, which are assumed to exist.
\end{theorem}
\begin{definition}
An interest rate model $\mathcal M$ is \textbf{consistent} with the forward rate manifold $\mathcal G$ if the consistent drift and
volatility conditions (\ref{drift-condition})-(\ref{volatility-condition}) hold.
\end{definition}

\section{THE HO-LEE SHORT RATE MODEL}
The short rate model proposed by Ho and Lee \cite{hl}, (henceforth HL) has a dynamic given by the $r$-process
\begin{equation}
\label{HoLee}
    dr(t)=\theta (t) dt+ \sigma dW(t).
\end{equation}
In (\ref{HoLee}), $\{W(t),\,t \ge 0\}$ is a Wiener process,  $\sigma>0$ and $\theta$ is a deterministic function. The HL model belongs to the family of affine short rate models. That is, if $P(t,T)$ denotes the price at time $t$ of a zero coupon bond with maturity $T$, then the term structure of the interest rate is given by:
\begin{equation}
    \label{affine}
    P(t,T)=e^{A(t,T)-r(t)B(t,T)},
    \end{equation}
for certain functions $A$ and $B$. In particular, in the case of the Ho-Lee model, $A$ and $B$ are given by:
\begin{equation}\label{B-HoLee}
    B(t,T)=T-t
\end{equation}
\begin{equation}\label{A-HoLee}
    A(t,T)=\int_t^T \theta(s) (s-T)ds + \frac{\sigma^2 (T-t)^3}{6}
\end{equation}
(see for instance \cite{brigo}). 
The forward rate curve is related with the term structure by the equation:
\begin{equation}\label{forward-curve}
    f(t,T)=-\frac{\partial \ln(P(t,T))}{\partial T}= -\frac{\partial A(t,T)}{\partial T}+ \frac{\partial B(t,T)r(t)}{\partial T}.
\end{equation}
Then, in this case replacing $A$ and $B$ by the expressions in (\ref{A-HoLee}) and (\ref{B-HoLee}), we get: 
\begin{eqnarray}\label{forward-HoLee}
f(t,T)&=& \frac{\partial}{\partial T}\left(-\int_t^T \theta(s) (s-T)ds  - \frac{\sigma^2 (T-t)^3}{6}\right) + r(t)\frac{\partial}{\partial T}(T-t) \notag\\
 &=&-\left(\theta(T) T-\int_t^T \theta(s)ds-T \theta(T)\right) - \frac{\sigma^2}{2} (T-t)^2+r(t) \notag\\
 &=&\int_t^T \theta(s)ds - \frac{\sigma^2}{2} (T-t)^2+r(t).
 \end{eqnarray}
 In particular it holds that $f(t,t)=r(t)$. If $T\mapsto f^*(0,T)$ is the observed initial forward curve, and $\theta$ is defined as
\begin{equation*}
    \theta(t)=\sigma^2 t + \frac{\partial f^*}{\partial T} (0, t),
\end{equation*}
then $f(0,T)=f^*(0,T)$. That is, the model parameters can be adjusted such that the initial forward curve fits the observed one.

 We now assume that the initial forward curve is given by a Nelson and Siegel parametric curve. That is, we define:
\begin{equation}\label{curva-inicial}
    f^*(0,T)=z_1+z_2 e^{-\lambda T} + z_3  T e^{-\lambda T}, \qquad T  \ge 0 ,
\end{equation}
where $z_1$, $z_2$, $z_3$ and $\lambda $ are fixed real numbers, $\lambda > 0$.  We want to study the evolution of this initial curve in the $t$ variable. With this particular choice of $f^*$, $\theta$ is given by:
\begin{equation*}
    \theta (t)=\sigma^2 t + (z_3- z_2 \lambda)e^{-\lambda t} - z_3 \lambda t e^{-\lambda t}
\end{equation*}
and the solution of the Ho-Lee stochastic differential equation (\ref{HoLee}) is given by:
\begin{align*}
    r(t)=&r(0)+\frac{\sigma^2 t^2}{2} - \frac{z_3-z_2 \lambda}{\lambda} e^{-\lambda t} - \frac{z_3}{\lambda} (1-e^{-\lambda t} (\lambda t+1)) + \sigma W(t)\\
    =&r(0)+\frac{\sigma^2 t^2}{2} -\frac{z_3}{\lambda} + [z_2+z_3t]e^{-\lambda t} + \sigma W(t),
\end{align*}
where $r(0)$ is the short rate value at time $t=0$. 

We compute the integral term in (\ref{forward-HoLee}):
\begin{eqnarray*}
\int_t^T \theta(s)ds&=&\int_t^T \sigma^2 s + (z_3- z_2 \lambda)e^{-\lambda s} - z_3 \lambda t e^{-\lambda s}ds \\
&=&\frac{\sigma^2}{2}(T^2-t^2)+ (z_2+z_3T)\, e^{-\lambda T}-(z_2+z_3t)\,e^{-\lambda t}.
\end{eqnarray*}
We can now derive an explicit formula for the forward rate curve process:
\begin{eqnarray}
f(t,T)&=&\frac{\sigma^2}{2}(T^2-t^2)+ (z_2+z_3T)\, e^{-\lambda T}-(z_2+z_3t)\,e^{-\lambda t} - \frac{\sigma^2}{2} (T-t)^2+r(t)\notag\\
&=&\sigma^2\,t(T-t)+ (z_2+z_3T)\, e^{-\lambda T}-(z_2+z_3t)\,e^{-\lambda t} +r(t).\label{HoLee-final}
\end{eqnarray}
The above computations allow us to state the following theorem.
\begin{theorem}\label{teorema-HoLee}
Let $r$ denote the HL short rate with a dynamic as stated in (\ref{HoLee}). Then, if the initial forward rate curve is a Nelson and Siegel parametric curve as in (\ref{curva-inicial}), the corresponding forward rate curve at time $t$ is given by the formula
$$f(t,T)= r(t)+\sigma^2\,t(T-t)+ (z_2+z_3T)\, e^{-\lambda T}-(z_2+z_3t)\,e^{-\lambda t} .$$
Let $\tau=T-t$ the time up to maturity. Using the Brace and Musiela parametrization, we denote $f_{HL}$ the forward curve given by $f_{HL}(t,\tau)=f(t,t+\tau)$. Then the forward rate curve $f_{HL}$ has the expression
   \begin{equation}
       \label{fMusiela-HL}
f_{HL}(t,\tau)= \sigma^2 t \tau + C_1(t)+ C_2(t) e^{-\lambda \tau} + C_3(t) \tau e^{-\lambda \tau} , \qquad t\ge 0, \ \tau  \ge 0 ,
   \end{equation}   
where $C_1$, $C_2$, $C_3$ are coefficients that depends on $t$ and the Nelson and Siegel parameters:
\begin{align*}
    C_1(t)&=r(t)-(z_2+z_3t) e^{-\lambda t}\\
    C_2(t)&=(z_2 + z_3 t) e^{-\lambda t}\\
    C_3(t)&=z_3 e^{-\lambda t}   
\end{align*}

\end{theorem}
\begin{proof}
The proof follows arranging terms after replacing $T$ by $\tau+t$ in equation (\ref{HoLee-final})
\end{proof}

The expression of the function $\tau \mapsto f_{HL}(t,\tau)$ given in the formula (\ref{fMusiela-HL}) is a sum of a linear function plus a Nelson and Siegel parametric curve. 
\begin{definition}
Let $\lambda>0$ and $g:[0,\infty)\mapsto \mathbb R$ be a function defined as
$$g(\tau)=z_0\tau+z_1+z_2e^{-\lambda\tau}+z_3e^{-\lambda\tau},$$
where $z_0$, $z_1$ and $z_2$ are constant real numbers. We call $g$ an \textit{linearly extended Nelson-Siegel} curve.
\end{definition}

 In particular, in the following subsection we study the consistency of the Ho-Lee model with a family of forward curve manifolds $\mathcal G^\lambda$ generated by linearly extended Nelson-Siegel curves.

\subsection{Consistency between the HL model and forward curve manifolds}

The forward rate curve $f$ given in (\ref{HoLee-final}) satisfies  the Heath, Jarrow and Morton drift condition:
$$\alpha(t,T)=\sigma(t,T)\int_{t}^T \sigma(t,s)\,ds=\sigma^2(T-t),$$ 
and the stochastic differential equation
$$df(t,T)=\alpha(t,T)\,dt+\sigma(t,T)\,dW(t)=\sigma^2(T-t)\,dt + \sigma \,dW(t),$$
with $T  \ge 0 $ and $0 \le t \le T$.
In terms of the Brace and Musiela parametrization, $f_{HL}$ satisfies:
\begin{equation}
df_{HL}(t,\tau)=\left(\frac{\partial}{\partial \tau}f_{HL}(t,\tau)+\sigma^2\,\tau\,\right)dt + \sigma \,dW(t), \label{Musiela-HoLee}
\end{equation}
$t  \ge 0 $, $\tau  \ge 0 $. Starting from equations (\ref{fMusiela-HL}) and (\ref{Musiela-HoLee}), our conjecture is that there exists a forward curve manifold containing linearly extended Nelson-Siegel curves that is consistent with the HL model. In fact, this is stated and proved by the following theorem.

\begin{theorem}\label{curvas-Ho-Lee}
Let $\mathcal Z= \mathbb R^4$, $\lambda>0$ and  
$$G^\lambda(\tau;\beta)=\beta_0\,\tau+\beta_1+\beta_2e^{-\lambda\tau}+\beta_3\tau e^{-\lambda\tau}, \qquad \tau  \ge 0 .$$ 
Let $f_{HL}(t,\cdot)$ be the Ho-Lee forward rate process. Then, for each $\lambda>0$ the forward curve manifold $\mathcal G^\lambda$ is $f_{HL}$-invariant. 
\end{theorem}
\begin{proof}

We shall apply Theorem \ref{theorem-consistency} to see that $\mathcal G^\lambda$ is $f_{HL}$-invariant. Because the volatility term in (\ref{Musiela-HoLee}) is deterministic, the standard differential equation is the same for the Ito and the Stratonovich integral formulation. 
The Frechet derivatives of $G^\lambda$ are given by: 
    \begin{eqnarray*}
        G^\lambda_\beta(\tau, \beta)&=& [ \tau, 1, e^{-\lambda \tau}, \tau e^{-\lambda \tau}] 
\\    
    G^\lambda_\tau(\tau,\beta)&=& \beta_0 + (-\beta_2 \lambda+ \beta_3) e^{-\lambda \tau} - \beta_3 \lambda \tau e^{-\lambda \tau} 
\end{eqnarray*}
In order to prove that $\mathcal G^\lambda$ is $f_{HL}$-invariant, we must check the drift and volatility consistency conditions (\ref{drift-condition}) and (\ref{volatility-condition}). We shall first prove that  $G^\lambda_\tau(\cdot,\beta)+ \sigma^2 \,(\cdot) \in \mbox{Im}[G^\lambda_\beta(\cdot,\beta)]$. This means that there must be real numbers $A$, $B$, $C$ and $D$ such that:
\begin{equation*}
    \beta_0 + (-\beta_2 \lambda+ \beta_3) e^{-\lambda \tau} - \beta_3 \lambda \tau e^{-\lambda \tau}  + \sigma^2 \tau = A \tau + B + C e^{-\lambda \tau} + D \tau e^{-\lambda \tau}.
\end{equation*}
In fact, this is possible taking
\begin{equation*}
    A=\sigma^2 ,\quad B=\beta_0, \quad  C=-\beta_2 \lambda+ \beta_3 \quad\mbox{and} \quad D= - \beta_3 \lambda,
\end{equation*}
so condition (\ref{drift-condition}) is satisfied. To prove condition (\ref{volatility-condition}), we must find $A$, $B$, $C$ and $D$ such that
\begin{equation*}
     \sigma = A \tau + B + C e^{-\lambda \tau} + D \tau e^{-\lambda \tau},
\end{equation*}
and this can be done taking $A=B=D=0$ and $B=\sigma$.
\end{proof}
Theorem \ref{curvas-Ho-Lee} implies that the forward curve manifold $G^\lambda$ is $f_{HL}$-invariant, so the next corollary follows:
\begin{corollary}
\label{corollary-HL}
For every $\lambda>0$,  the forward curve manifold $G^\lambda$ is consistent with the Ho-Lee short rate model.
\end{corollary}
Corollary \ref{corollary-HL} implies that, in the particular case that 
$\theta$ in (\ref{HoLee}) is chosen such that the initial forward rate curve fits the Nelson and Siegel curve 
$$f^*(0,\tau)=z_1+z_2e^{-\lambda\tau}+z_3\tau e^{-\lambda\tau}, \qquad \tau  \ge 0 ,$$
then, for each $t  \ge 0 $ the corresponding forward rate curve $f_{HL}(t,\tau)$ can be written as a linearly extended Nelson and Siegel curve. That is:
$$f_{HL}(t,\tau)=\beta_0\tau+\beta_1+\beta_2e^{-\lambda\tau}+\beta_3\tau e^{-\lambda\tau}, \qquad \tau  \ge 0 ,$$
with $\beta_0$, $\beta_1$, $\beta_2$, $\beta_3$ depending only on $t$ and $r(t)$.

\begin{remark}
It must be noted that the manifold $G^\lambda$ does not  contain every Nelson and Siegel curve, but only those with the exponential term equal to $e^{-\lambda\tau}$, $\tau  \ge 0 $. If one considers a wider manifold $\mathcal G$ with parameter set $\mathcal Z=\mathbb R^5$ such that
$$G(\tau;\beta)=\beta_0\,\tau+\beta_1+\beta_2e^{-\beta_4\tau}+\beta_3\tau e^{-\beta_4\tau}, \qquad \tau  \ge 0 ,$$
then the drift consistency condition is not satisfied, because it  requires $\beta_4=\lambda$.
\end{remark}

\section{THE HULL AND WHITE MODEL}
The short rate model proposed by Hull and White \cite{hw}, (henceforth HW) or extended Vasicek model has the following stochastic differential equation:
\begin{equation}
\label{HullWhite}
    dr(t)=(\theta(t) - a r(t))dt + \sigma dW(t)
\end{equation}
where $a$, $\sigma$ are positive real numbers and $\theta$ is a deterministic function. $\theta$ can be chosen in such a way that the initial forward curve $f(0,\cdot)$ fits with the observed data at $t=0$. Let $f^*(0,\cdot)$ be a particular forward curve observed at $t=0$. Then $\theta$ is defined as: 
\begin{equation}\label{theta-HW}
    \theta(t)=\frac{\partial f^*}{\partial T}(0,t) + a f^*(0,t)+\frac{\sigma^2}{2a}(1-e^{-2at})
\end{equation}
The Hull and White short rate model belongs also to the class of affine models, and in this case the functions $A$ and $B$ in equation (\ref{affine}) are given by \cite{brigo}:
\begin{eqnarray}
    B(t,T)&=&\frac{1}{a}(1-e^{-a(T-t)}) \label{A-HW}\\
    A(t,T)&=&\int_t^T \left(\frac{1}{2}\sigma^2 B^2(s,T) - \theta(s) B(s,T)\right) ds\label{B-HW}
\end{eqnarray}
The corresponding forward rate curve is given by:
\begin{eqnarray}
    f(t,T)&=& -\frac{\partial A(t,T)}{\partial T}  + \frac{\partial B(t,T) r(t)}{\partial T}  \notag\\
          &=& -\frac{\partial}{\partial T}\left(\int_t^T \left(\frac{1}{2}\sigma^2 B^2(s,T) - \theta(s) B(s,T)\right) ds\right) + r(t)e^{-a(T-t)}\label{forward-HW}
\end{eqnarray}
Now, by Leibniz rule, we have:
\begin{eqnarray}\
    \frac{\partial}{\partial T}\left(\int_t^T \frac{1}{2}\sigma^2 B^2(s,T) ds\right)&=&  \frac{1}{2}\sigma^2 B^2(T,T)+ \int_t^T \frac{1}{2}\sigma^2 \frac{\partial B^2(s,T)}{\partial T} ds \notag\\
    &=&\sigma^2 \int_t^T B(s,T)\frac{\partial B(s,T)}{\partial T}\,ds\notag\\
    &=&\sigma^2 \int_t^T \frac{(1-e^{-a(T-s)})}{a}\,e^{-a(T-s)}\,ds\notag\\
    &=&-\frac 12 \,\sigma^2\, B(t,T)^2=-\frac {\sigma^2}{2a^2} \, (1-e^{-a(T-t)})^2 \label{first-term}
\end{eqnarray}
\begin{eqnarray*}
    \frac{\partial}{\partial T}\left(\int_t^T \theta(s) B(s,T)\right)&=& \theta(T)B(T,T)+\int_t^T \theta(s)\frac{\partial B(s,T)}{\partial T}\,ds\\
&=& \int_t^T \theta(s)e^{-a(T-s)}\,ds.    
\end{eqnarray*}
We now assume that the initial forward curve is fitted to a Nelson and Siegel parametric curve,
\begin{equation*}
    f^*(0,T)=z_1+z_2 e^{-\lambda T}+z_3 T e^{-\lambda T}, \qquad T  \ge 0 
\end{equation*}
Then, the solution of (\ref{HullWhite}) and the function $\theta$ are given by the following expressions:
\begin{eqnarray*}
    r(t)&=&r(0) e^{-at} + \int_0^t e^{-a(t-u)} \theta(u) du+ \sigma \int_0^t e^{-a(t-u)} dW(u)\\
    &=&r(0) e^{-at} + \alpha(t)-\alpha(0) e^{-at}+ \sigma \int_0^t e^{-a(t-u)} dW(u)\\
\theta(t)&=& az_1+ (z_3-z_2\lambda+az_2)e^{-\lambda t} +(az_3 -z_3\lambda)t e^{-\lambda t}+\frac{\sigma^2}{2a}(1-e^{-2at})
\end{eqnarray*} where
\begin{equation*}
    \alpha(t)=f^*(0,t)+\frac{\sigma^2}{2a^2}(1-e^{-at})^2.
\end{equation*}
Now we can compute explicitly
\begin{eqnarray}
\int_t^T \theta(s)e^{-a(T-s)}\,ds &=& \alpha(T)- \alpha(t) e^{-a(T-t)} \notag\\
&=&f^*(0,T)+ \frac{\sigma^2}{2a^2} (1-e^{-aT})^2 -\alpha(t)e^{-a(T-t)}  \label{second-term}
\end{eqnarray}

Replacing the expressions (\ref{first-term}) and (\ref{second-term}) in the forward rate curve formula (\ref{forward-HW}), we get:
\begin{equation*}
    f(t,T)=-\frac{\sigma^2}{2a^2} (1-e^{-a(T-t)})^2 + f^*(0,T)+ \frac{\sigma^2}{2a^2} (1-e^{-aT})^2 -\alpha(t)e^{-a(T-t)} + r(t)e^{-a(T-t)}
\end{equation*}
With the above computations we arrive to the next theorem.
\begin{theorem}\label{teorema-HullWhite}
Let $r$ denote the Hull and White short rate model with the dynamics stated in (\ref{HullWhite}). Then, if the initial forward rate curve is a Nelson and Siegel parametric curve as in (\ref{curva-inicial}), the corresponding forward rate curve at time $t$ is given by the formula
\begin{equation*}
    f(t,T)=-\frac{\sigma^2}{2a^2} (1-e^{-a(T-t)})^2 + f^*(0,T)+ \frac{\sigma^2}{2a^2} (1-e^{-aT})^2 -\alpha(t)e^{-a(T-t)} + r(t)e^{-a(T-t)}.
\end{equation*}
Let $\tau=T-t$ the time up to maturity. Using the Brace and Musiela parametrization, we denote $f_{HW}$ the forward curve given by $f_{HW}(t,\tau)=f(t,t+\tau)$. Then the forward rate curve $f_{HW}$ has the expression
   \begin{equation}\label{fMusiela-HW}
    f_{HW}(t,\tau)=C1(t) e^{-a\tau} + C2(t) e^{-2a\tau} + C3(t)+ C4(t) e^{-\lambda \tau} + C5(t)\tau e^{-\lambda \tau} 
\end{equation}
where $C_1$, $C_2$, $C_3$, $C_4$ and $C_5$ are coefficients that depends on $t$, $r(t)$ and the Nelson and Siegel curve parameters:
$$C_1(t)=\frac{\sigma^2}{a^2}(1 - e^{-at}) - \alpha(t) +r(t), \qquad 
    C_2(t)=\frac{\sigma^2}{2a^2}  (e^{-2at}-1), $$
$$    C_3(t)=z_1, \qquad  C_4(t)=z_2 e^{-\lambda t} + z_3 t e^{-\lambda t}, \qquad  C_5(t)=z_3 e^{-\lambda t}
.$$
\end{theorem}

The expression of the function $\tau \mapsto f_{HW}(t,\tau)$ given in the formula (\ref{fMusiela-HW}) is a sum of an exponential function plus a Nelson and Siegel parametric curve. 
\begin{definition}
Let $\lambda>0$ and $g:[0,\infty)\mapsto \mathbb R$ be a function defined as
$$g(\tau)= c_1e^{-a\tau} + c_2e^{-2a\tau}+z_0\tau+z_1+z_2e^{-\lambda\tau}+z_3e^{-\lambda\tau},$$
with $c_1$, $c_2$, $z_0$, $z_1$ and $z_2$ constant real numbers. We call $g$ an \textit{exponentially extended Nelson-Siegel} curve.
\end{definition}
 In particular, in the following subsection we study the consistency of the Hull and White model with a family of forward curve manifolds $\mathcal G^\lambda$ generated by exponentially extended Nelson-Siegel curves.

\subsection{Consistency between the Hull and White model and forward curve manifolds}
The dynamics of the forward curve process is given in terms of the Ito integral formulation by 
$$
    df(t,T)= \frac{\sigma^2}{a} e^{-a (T-t)} (1-e^{-a(T-t)}) dt + \sigma e^{-a(T-t)} dW(t), 
$$
$0\le t<T<\infty$, and in the Musiela parametrization $\tau=T-t$,
\begin{equation}
\label{HullWhite-Musiela}
     df_{HW}(t,\tau)=\left(\frac{\partial}{\partial \tau}f_{HW}(t,\tau)+ \frac{\sigma^2}{a} e^{-a \tau} (1-e^{-a\tau}) \right)dt + \sigma e^{-a\tau} dW(t)
\end{equation}
$\tau  \ge 0 $, $t\ge 0$. 
Because the volatility term is a deterministic function, the Stratonovich and the Ito integral formulation are the same. 

We state the next theorem:
\begin{theorem}\label{curvas-HullWhite}
Let $\mathcal Z= \mathbb R^5$ and for each $\lambda>0$ let $G^\lambda:\mathcal Z \mapsto C(0,\infty)$ be defined as
   \begin{equation}
       \label{G-HullWhite}
   G^\lambda(\tau;\beta)=\beta_1 e^{-a\tau} + \beta_2 e^{-2a \tau} + \beta_3+ \beta_4 e^{-\lambda \tau} + \beta_5 \tau e^{-\lambda \tau}.
   \end{equation}
Let $f_{HW}(t,\cdot)$ be the Hull and White forward rate process given by the standard differential equation (\ref{HullWhite-Musiela}). Then, the forward curve manifold $\mathcal G^\lambda$ is $f_{HW}$-invariant. 
\end{theorem}
\begin{proof}
We note that the Frechet derivatives $G^\lambda_\beta$ and $G^\lambda_\tau$ are given by:
\begin{eqnarray*}
G^\lambda_{\beta}(\tau; \beta)&=& [ e^{-a\tau},e^{-2a \tau}, 1, e^{-\lambda \tau}, \tau e^{-\lambda \tau}] \\
G^\lambda_{\tau}(\tau;\beta)&=&   \beta_1 (-a) e^{-a\tau}+ \beta_2 (-2a) e^{-2a\tau} + (-\beta_4 \lambda+ \beta_5) e^{-\lambda \tau} - \beta_5 \lambda \tau e^{-\lambda \tau}.
\end{eqnarray*}
So, first we prove that
$
    G^\lambda_{\tau}(\cdot,\beta)+ \frac{\sigma^2}{a} e^{-a(\cdot)}(1-e^{-a (\cdot)} ) \in \mbox{Im}[G^\lambda_{\beta}(\cdot,\beta)]
$. So we look for real numbers $A$, $B$, $C$, $D$ and $E$ such that 
\begin{align*}
    \beta_1 (-a) e^{-a\tau}+ \beta_2 (-2a) e^{-2a\tau} + (-\beta_4 \lambda+ \beta_5) e^{-\lambda \tau} - \beta_5 \lambda \tau e^{-\lambda \tau}+ \frac{\sigma^2}{a} e^{-a\tau}(1-e^{-a \tau} ) \\
    = A e^{-a\tau}+ B e^{-2a \tau}+ C+ D e^{-\lambda \tau}+E \tau e^{-\lambda \tau}.
\end{align*}
This is true setting
\begin{align*}
    A=-a \beta_1+ \frac{\sigma^2}{a} ,\;\; B=-2a \beta_2 - \frac{\sigma^2}{a} ,\;\; C=0, \;\; D=-\beta_4 \lambda+ \beta_5 , \;\; E= - \beta_5 \lambda.
\end{align*}
We next prove that
$
     \sigma e^{-a(\cdot)}\in Im[G_{\beta}(\cdot,\beta)]
$, or equivalently, we look for real numbers such that 
\begin{equation*}
     \sigma e^{-a \tau} = A e^{-a\tau}+ B e^{-2a \tau}+ C+ D e^{-\lambda \tau}+E \tau e^{-\lambda \tau}.
\end{equation*}
Setting $A=\sigma$ and $B=C=D=E=0$ the identity follows.

\end{proof}
As a conclusion, we have the following corollary:
\begin{corollary}
For each $\lambda>0$, the Hull and White short rate model is consistent with the forward curve manifold $\mathcal G^\lambda$ given in (\ref{G-HullWhite}).
\end{corollary}
\begin{remark}
As in the case of the HL-model, we see that if the initial forward curve is a Nelson-Siegel curve, then the following forward curves are exponentially extended Nelson-Siegel curves, belonging to a particular manifold $\mathcal G^\lambda$. It must be noted that, even when the forward curve process moves on a wider manifold $\mathcal H$ with parameter set $\mathcal Z=\mathbb R^6$,
$$H(\tau;\beta)=\beta_1 e^{-a\tau} + \beta_2 e^{-2a \tau} + \beta_3+ \beta_4 e^{-\beta_6 \tau} + \beta_5 \tau e^{-\beta_6 \tau}, \quad \tau  \ge 0 ,$$
it is not true that this particular manifold is consistent with the HW model. The parameter $\beta_6$ must be equal to the parameter $\lambda$ in (\ref{nelson-siegel}) to get the drift consistency condition (\ref{drift-condition}).
\end{remark}

\section{CONCLUSION}

In the previous sections we presented two families of parametric curves that are consistent with the Ho-Lee and the Hull and White short rate models, respectively. These families contain curves that are extensions of the Nelson and Siegel classical curves, in the sense that each element can be written as the sum of a Nelson and Siegel curve plus a linear function in $\tau$ or plus an exponential function in $\tau$. We proved that each of these short rate models are consistent with a family of forward curve manifolds $\mathcal G^\lambda$. Moreover, for each $\lambda$, $G^\lambda$ is a set of forward curves drived by four factors. Three of them are the Nelson and Siegel factors $1$, $e^{-\lambda \tau}$ and $\tau \,e^{-\lambda \tau}$. In the HL model, the fourth factor is a linear function in $\tau$ multiplied by $\sigma$, the short rate volatility. In the HW model it is an exponential decaying function $c_1e^{-a\tau}+c_2e^{-2a\tau}$, where $a$ is the model parameter. In both cases, the \textit{shape} parameter $\lambda$ remains the same along the forward rate process.

The Ho-Lee and Hull and White short rate models are within the class of arbitrage free models. That means that the initial forward rate curve can be fitted to any observed curve. In equilibrium models with an affine term structure of interest rate, such as Vasicek and CIR, it is not possible to choose the initial curve as a Nelson-Siegel one, and this happens because the drift term has constant parameters. Nevertheless, our conjecture and further research work, is that it is possible to define an algebraically method to find the parameters of a Nelson-Siegel forward curve that is the closest to the model forward rate curve.

\bibliography{Bibliography}

\begin{thebibliography}{1}

\bibitem{jan}
J.~Annaert, A.~G.~P. Claes, de~Ceuster, and H.~M.~J.~K., Zhang.
\newblock Estimating the {Y}ield {C}urve {U}sing the {N}elson-{S}iegel {M}odel:
  {A} {R}idge {R}egression {A}pproach.
\newblock {\em International {R}eview of {E}conomics and {F}inance},
  (27):482--496, 2013.

\bibitem{bc}
T.~Björk and B.~Christensen.
\newblock Interest rate dynamics and consistent forward rate curves.
\newblock {\em Mathematical Finance}, 9(4):323–348, 1999.

\bibitem{bm}
A.~Brace and M.~Musiela.
\newblock A multi factor {G}auss {M}arkov implementation of {H}eath, {J}arrow
  and {M}orton.
\newblock {\em Mathematical Finance}, 9:563--576, 1994.

\bibitem{brigo}
D.~Brigo and F.~Mercurio.
\newblock {\em Interest {R}ate {M}odels- {T}heory and {P}ractice}.
\newblock Springer Verlag, Berlin, 2006.

\bibitem{hjm}
D.~Heath, R.~Jarrow, and A.~Morton.
\newblock Bond pricing and the term structure of interest rates.
\newblock {\em Econometrica}, 60:77--106, 1992.

\bibitem{hl}
T.~Ho and S.~Lee.
\newblock Term structure movements and pricing interest rate contingent claims.
\newblock {\em Journal of Finance}, 41:1011--1029, 1986.

\bibitem{hw}
J.~Hull and A.~White.
\newblock Pricing interest-rate-derivative securities.
\newblock {\em The review of Financial Studies}, 3:573--592, 1990.

\bibitem{ns}
C.~Nelson and A.~Siegel.
\newblock Parsimonious modeling of yield curves.
\newblock {\em Journal of Business}, 60:473--489, 1998.

\bibitem{oks}
B.~Øksendal.
\newblock {\em Stochastic differential equations}.
\newblock Springer Verlag, Berlin, 1998.

\end{thebibliography}
\bibliographystyle{plain}

\end{document}